\documentclass{article}
\usepackage{graphicx} 
\usepackage{xcolor}
\usepackage[numbers,sort&compress]{natbib}
\usepackage{geometry}
\usepackage{amsmath}
\usepackage{amsthm}
\usepackage{tikz}
\usepackage{amssymb}
\usepackage{braket}
\usepackage{appendix}
\usepackage{hyperref}
\usepackage{authblk}

\theoremstyle{plain}
\newtheorem{theorem}{Theorem}[section]

\newtheorem{proposition}[theorem]{Proposition}

\theoremstyle{definition}

\theoremstyle{remark}

\newcommand{\smallbox}[4]{
\pgfmathsetmacro{\x}{#1}
\pgfmathsetmacro{\y}{#2}
\pgfmathsetmacro{\rx}{#3}
\pgfmathsetmacro{\ry}{#4}

\begin{scope}[shift={(\x,\y)},xscale=\rx,yscale=\ry]
\fill[white](1,-1)--(1,1)--(-1,1)--(-1,-1)--(1,-1);
\draw(1,-1)--(1,1)--(-1,1)--(-1,-1)--(1,-1);
\end{scope}

}

\begin{document}
\title{Resonance Breaking in Noisy Shor's Algorithm}
\author[1,2]{Zhengwei Liu\thanks{%
  Email: \href{mailto:liuzhengwei@mail.tsinghua.edu.cn}
  {\nolinkurl{liuzhengwei@mail.tsinghua.edu.cn}}%
}}
\affil[1]{Yau Mathematical Sciences Center and Department of Mathematics, Tsinghua University, Beijing, 100084, China}
\affil[2]{Yanqi Lake Beijing Institute of Mathematical Sciences and Applications, Huairou District,
Beijing, 101408, China}
\date{\today}

\maketitle


\begin{abstract}
We argue that the exponential quantum advantage in Shor's algorithm is broken under one-layer depolarizing noise of arbitrary small error rate, by analyzing a resonance breaking phenomenon in noisy quantum circuits. First, we express the distribution of measurements on $n$-qubit strings as the superposition of the $4^n$ wave functions in the Pauli path integral. In the noiseless case, the bit-strings achieving resonant peaks guaranteed a constant rate of successful measurements to factor a large number. Secondly, when the middle layer of the quantum circuit has independent depolarizing noise of rate $\lambda$, for any Hamming weight $d$, we obtain corresponding measurement rate bounded by $2(1-\lambda)^d$ for higher frequency terms and by $O(n^d)/2^{n/2}$ for low frequency terms. The rate approaches to zero as $n$ and $d$ approaches to infinity. Resonance breaking destroys the exponential quantum advantage. Furthermore, we design a classical factoring algorithm in polynomial time $O(n^d)$ to substitute the low frequency contribution.  
\end{abstract}

\section{Introduction}

Shor's factoring algorithm sparked widespread interest in quantum computing.
By establishing polynomial-time integer factorization on an ideal quantum computer~\cite{Shor1997}, it offered a concrete prospect of an exponential computational advantage for a problem of direct cryptographic significance~\cite{RivestShamirAdleman1978}.
Realizing this advantage in practice has become a major goal of quantum computing research~\cite{GidneyEkera2021,GouzienSangouard2021,Gidney2025,WebsterEtAl2026,CainEtAl2026}.
Shor's algorithm relies on quantum amplitudes adding constructively to produce peaks in the measurement distribution that encode the order.
Noise can disrupt the relative phases of these amplitudes, weakening their constructive interference, and reducing the probability of obtaining outcomes near the peaks. Recently, Cai proved that random angular errors in Fourier gates can make the probability of useful order-finding results negligible, preventing efficient factorization even when noise strength decreases with input size~\cite{Cai2023}. Early studies investigated how decoherence and dissipation affect the performance of Shor's algorithm~\cite{ChuangEtAl1995,MiquelEtAl1996,BarencoEtAl1996}.


Inspired by the Quon classical simulation framework, where the number of remaining ``Magic holes'' governs the cost of its classical simulation algorithm~\cite{FengEtAl2025}, we observed that quantum advantage comes from composing different types of quantum circuits, such as Clifford circuits~\cite{AaronsonGottesman2004} or matchgate circuits~\cite{Valiant2002,JozsaMiyake2008}. In contrast, if the composition boundary of different types of circuits has noise, then the noise can severely degrade the exponential quantum advantage of the quantum circuit.

We analyze this phenomenon in Shor's algorithm to factor a large number $N$ for a given number $a$ with order $r$ mod $N$. More precisely, we apply independent single-qubit depolarizing noise of rate $\lambda$ to the $n$-qubit control register between modular exponentiation and the inverse Fourier transform, with all other operations ideal. First, we present the resonance effect of exponentially many wave functions in the Pauli path integral~\cite{GaoDuan2018,AharonovEtAl2023,ShaoEtAl2024,FontanaEtAl2025,GonzalezGarciaEtAl2025,MartinezEtAl2025,ShaoEtAl2025} for Shor's quantum circuit. The probability of measuring a bit-string $x$ is
\begin{equation}
P_{\lambda}(x)= 2^{-2n} \sum_{0\leq u,v < 2^n ; r \mid u-v}
(1-\lambda)^{h(u,v)} \omega^{x (u-v)}  ,
\end{equation}
where $\omega=e^{\frac{2\pi i}{2^n}}$ and $h(u,v)$ is the Hamming distance between the binary expansions of $u$ and $v$.
For the case without noise, namely $\lambda=0$, the $4^n$ phases $\omega^{x (u-v)}$, for $0\leq u,v < 2^n$, achieves resonance peaks at $x$ in the following set
\begin{equation}
\mathcal{E}_r:=\left\{x\in\{0,\ldots,2^n-1\}\ \middle|\
\exists k\in\{0,\ldots,r-1\},\
\left|\frac{x}{2^n}-\frac{k}{r}\right|
\leq\frac{1}{2^{n+1}}\right\}.
\end{equation}
Then applying the continued fraction theorem to such $\frac{x}{2^n}$ one can find the order $r$ and then a non-trivial factor of the large number $N$. 
The successful measurement rate is guaranteed by the crucial lower bound 
\begin{equation}
\sum_{x \in \mathcal{E}_r}
P_0(x) \geq \frac{4}{\pi^2}-O(\frac{1}{N}).
\end{equation}

Secondly, we show that the resonance breaks when the noise rate $\lambda>0$. For any Hamming weight $d$, the high frequency terms $h(u,v)\geq d$ will diminish exponentially proportion to $(1-\lambda)^d$, and the low frequency terms have a bound $O(n^d)/2^{n/2}$. 
\begin{equation}
\sum_{x \in \mathcal{E}_r}
P_{\lambda}(x) \leq \frac{O(n^d)}{2^{n/2}}+2(1-\lambda)^d.
\end{equation}
As $n\to \infty$ and $d\to \infty$, it approaches to zero. The successful measurement rate is no longer guaranteed. The resonance breaking phenomenon destroys the exponential quantum advantage.

For a fixed $n$, we further present a classical algorithm in polynomial time $O(n^d)$ to substitute the low frequency contributions. To keep the same exponential factor $(1-\lambda)^d$, when the noise rate $\lambda$ decreases, the degree $d$ increases. 
Quantum error correction~\cite{Shor1995,Steane1996,DennisEtAl2002} may help improve performance, and the threshold theorem allows arbitrarily small logical error rates below the threshold under suitable noise assumptions~\cite{AharonovBenOr2008,KnillLaflammeZurek1998,AliferisGottesmanPreskill2006}, achieving this suppression requires additional physical qubits and repeated correction operations~\cite{FowlerEtAl2012,GoogleQuantumAI2025,BravyiEtAl2024}. In practice, the polynomial degree may grow too high to be achieved by a classical computer, while reducing the noisy rate using quantum error corrections on a quantum computer remains challenging as well~\cite{Terhal2015}.

We hope that our analysis of the frequency decomposition of resonant wave functions in the Pauli path integral will shed light on the theoretical understanding of exponential quantum advantage in quantum circuits. The resonance breaking phenomenon allows us to analyze a polynomial time classical simulation to a single noisy circuit, instead of the average behavior of random noisy circuits. Resonance seems essential to achieve a reasonable measurement rate, as every path only contributes to an exponentially small rate in the path integral. The higher frequency contribution may be compressed exponentially by noise and the lower frequency contribution may be substituted by a corresponding classical algorithm in polynomial time.

\section{Shor's algorithm}
We first recall Shor's algorithm for factoring large numbers~\cite{Shor1997}.
We write $0\leq x<2^n$ in binary $x=\sum\limits_{i=0}^{n-1} 2^ix_i$, denoted as $\ket{x}:=\ket{x_{n-1}...x_1x_0}$.
The quantum Fourier transform is defined as
\begin{equation}
QFT\ket{x}:=\frac{1}{\sqrt{2^n}}\sum_{y=0}^{2^n-1}\omega^{xy}\ket{y}, \quad\omega=e^{\frac{2\pi i}{2^n}},
\end{equation}
which can be implemented by a quantum circuit of elementary gates. 



\newcommand{\stringbox}[5]{
\pgfmathsetmacro{\rx}{#1}
\pgfmathsetmacro{\ry}{#2}
\pgfmathsetmacro{\xx}{#3}
\pgfmathsetmacro{\yy}{#4}
\pgfmathsetmacro{\rr}{#5}

\begin{scope}[shift={(\xx,\yy)}]
    \draw(-1*\rx-\rr,-\ry+\rr)--(-1*\rx-\rr,\ry-\rr)[bend left=30]to(-1*\rx+\rr,\ry+\rr)--(\rx-\rr,\ry+\rr)[bend left=30]to(\rx+\rr,\ry-\rr)--(\rx+\rr,-1*\ry+\rr)[bend left=30]to(\rx-\rr,-1*\ry-\rr)--(-\rx+\rr,-1*\ry-\rr)[bend left=30]to(-\rx-\rr,-1*\ry+\rr);
\end{scope}
}

Suppose $N$ is a large number, e.g., $N=pq$ for a pair of prime numbers in RSA cryptography. Take the qubit number $n$, such that 
\begin{equation}\label{Equ:n}
N^2< 2^n \leq 2 N^2.    
\end{equation}
Choose a number $1<a<N$, e.g. $a=2$, and first compute
$\gcd(a,N)$. If this gcd is nontrivial, it already gives a factor of
$N$. Otherwise, assume $\gcd(a,N)=1$. Shor established a quantum
circuit to find the multiplicative order $r$ of $a$ modulo $N$.
The order $r$ is called good, if $r$ is even and
$N\nmid a^{\frac{r}{2}}+1$ and $ N\nmid a^{\frac{r}{2}}-1$.
For a good $r$,
\begin{equation}
N \mid (a^{\frac{r}{2}}+1)(a^{\frac{r}{2}}-1),
\end{equation}
and both $gcd(a^{\frac{r}{2}}+1, N)$ and $gcd(a^{\frac{r}{2}}-1, N)$ are non-trivial factors of $N$. When $N$ is odd and has at least two distinct prime factors, repeated independent trials with uniformly random $a\in(\mathbb{Z}/N\mathbb{Z})^\times$ yield a good order with high probability~\cite{Shor1997}.

Shor applied the following quantum circuit to find the order $r$.

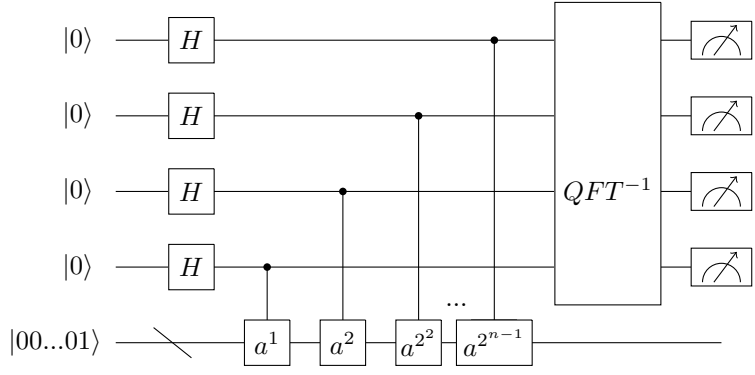
\begin{center}
    \begin{tikzpicture}
        \begin{scope}[yscale=-1]
        
        \draw(-3,0)--(5,0);
        \foreach \i in {1,2,3,4}{
        \draw(-3,-1*\i)--(5,-1*\i);
        \node at (-3.5,-1*\i){$\ket{0}$};
        \smallbox{-2}{-1*\i}{0.3}{0.3}
        \node at (-2,-1*\i){$H$};
        }
        \draw(-2.5,-.2)--(-2,+.2);
        \node at (-3.8,0){$\ket{00...01}$};

        \foreach \j in {2,1,0,-1}{
        \fill(\j,-2-\j) circle(.05);
        \draw(\j,0)--(\j,-2-\j);
        \smallbox{\j}{0}{0.3}{0.3}
        }
        
        \smallbox{2}{0}{0.5}{0.3}
        \node at (-1,0){$a^1$};
        \node at (0,0){$a^2$};
        \node at (1,0){$a^{2^2}$};
        \node at (1.5,-0.5){$...$};
        \node at (2,0){$a^{2^{n-1}}$};

        \smallbox{3.5}{-2.5}{0.7}{2}
        \node at (3.5,-2){$QFT^{-1}$};
\foreach \k in {-1,-2,-3,-4}{
        \smallbox{5}{\k}{0.4}{0.25}
        \begin{scope}[shift={(5,\k)},yscale=-1]
            \draw(-.25,-.2)arc[start angle=180,end angle=0,radius=.25];
            \draw[->](-.1,-.2)--(.2,.2);
        \end{scope}
}

        \end{scope}
    \end{tikzpicture}
\end{center}
The controlled transformation apply to bottom line as the multiplication by $a^{2^i}$ mod $N$.
We simplify the quantum circuit presentation as follows:
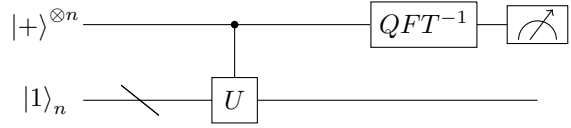
\begin{center}
    \begin{tikzpicture}
        \begin{scope}[yscale=-1]
        \foreach \i in {0,1}{
        \draw(-1,-1*\i)--(5,-1*\i);
        \draw(-.5,-.2)--(0,+.2);
        }
        
        \node at (-1.5,-1){$\ket{+}^{\otimes n}$};
        \node at (-1.5,0){$\ket{1}_n$};
        
        \fill(1,-1) circle(.05);
        \draw(1,0)--(1,-1);
        \smallbox{1}{0}{0.3}{0.3}
        \node at (1,0){$U$};

        \smallbox{3.5}{-1}{0.7}{0.3}
        \node at (3.5,-1){$QFT^{-1}$};

        \smallbox{5}{-1}{0.4}{0.25}
        \begin{scope}[shift={(5,-1)},yscale=-1]
            \draw(-.25,-.2)arc[start angle=180,end angle=0,radius=.25];
            \draw[->](-.1,-.2)--(.2,.2);
        \end{scope}

        \end{scope}
    \end{tikzpicture}
\end{center}

Considering $\ket{+}^{\otimes n}$ as a superposition of all $\ket{u}:=\ket{u_{n-1}...u_0}$, 
the state after $QFT^{-1}$ is
\begin{equation}
2^{-n}\sum_{x=0}^{2^n-1}\sum_{u=0}^{2^n-1}
\omega^{-xu}\ket{x}\otimes\ket{a^u\bmod N}.
\end{equation}
After tracing out the second register, the probability of measuring
$\ket{x}$ is
\begin{equation}
\begin{aligned}
P(x)
&=2^{-2n}\sum_{0\leq u,v<2^n;\ r\mid u-v}
\omega^{-x(u-v)}\\
&=2^{-2n}\sum_{0\leq u,v<2^n;\ r\mid u-v}
\omega^{x(u-v)},
\end{aligned}
\end{equation}
where the second equality follows by exchanging $u$ and $v$.

When $|\frac{x}{2^n}- \frac{k}{r}| \leq  \frac{1}{2^{n+1}}$ for some integer $k$, Eq.~\eqref{Equ:n} and $r<N$ give
\begin{equation}
\left|\frac{x}{2^n}-\frac{k}{r}\right|
\leq\frac{1}{2\cdot2^n}<\frac{1}{2r^2}.
\end{equation}
Then the continued fraction theorem implies that the reduced form of $\frac{k}{r}$ appears among the convergents of $\frac{x}{2^n}$, which can be computed in time polynomial in $n$~\cite{Shor1997,NielsenChuang2010}.
The denominator $r(x)$ is $\frac{r}{gcd(r,k)}$, where $gcd(r,k)$ is the greatest common divisor. The order $r$ can be tested through the condition $a^r\equiv 1$ mod $N$. If $r(x)=r$, then we obtain the order $r$. Otherwise, one can sample different measurement $x$, and choose multiple $r(x)$ for different $x$, and then test their least common multiple (lcm) as a candidate for the order $r$~\cite{Shor1997}.

Note that
\begin{equation}
P(x)= 2^{-2n} \sum_{0\leq u,v < 2^n ; r \mid u-v} \omega^{2^n(\frac{x}{2^n}-\frac{k}{r}) (u-v)}.
\end{equation}

When $\frac{x}{2^n}$ approaches $\frac{k}{r}$, the phase $\omega^{2^n(\frac{x}{2^n}-\frac{k}{r}) (u-v)}$ approaches 1. Then $P(x)$ as their resonant sum approaches approximately to $\frac{1}{r}$.
When $|\frac{x}{2^n}- \frac{k}{r}|$ increases from $0$ to $\frac{1}{2^{n+1}}$, the probability decreases from $\frac{1}{r}$ to $\frac{4}{\pi^2 r}$ up to a small error $O(\frac{1}{2^n})$ as shown in ~\cite{Shor1997,NielsenChuang2010}.
Summing these peak bounds yields the crucial constant bound for the effective measurements
\begin{equation}\label{Equ:Shor lower bound}
\sum_{\substack{0\leq x<2^n:\ \exists k\in\{0,\ldots,r-1\},\\
|\frac{x}{2^n}-\frac{k}{r}|\leq\frac{1}{2^{n+1}}}}
P(x) \geq \frac{4}{\pi^2}-O(\frac{1}{N}).
\end{equation}
Combined with repeated sampling, classical post-processing and random choices of $a$, this gives a constant overall success probability using polynomial time~\cite{Shor1997,Ekera2024}.

In Shor's factoring algorithm, the quantum circuit is to find the order $r$ of $a$ through measurements, in addition to the classical algorithm of the continued fraction theorem.

\section{Resonance breaking under one layer noise}

Now let us analyze noisy Shor's algorithm.
The depolarizing noise of a single qubit is a quantum channel
$\mathcal{N}_{\lambda}$ with noise rate $0\leq\lambda\leq 1$, defined by
\begin{equation}
\mathcal{N}_{\lambda}(D)=(1-\lambda)D+\lambda I/2
\end{equation}
for every single-qubit density matrix $D$. Thus $\lambda=0$ is the
noiseless case, while every traceless Pauli operator is contracted by
the coherence factor $1-\lambda$.
We add a one-layer depolarizing noise $\mathcal{N}_{\lambda}$ to the most fragile part of Shor's quantum circuit, which acts on the $n$ qubits between the controlled unitary and $QFT^{-1}$. Let $P_{\lambda}(x)$ be the probability of $n$-qubit measuring $x$, which is presented as follows:

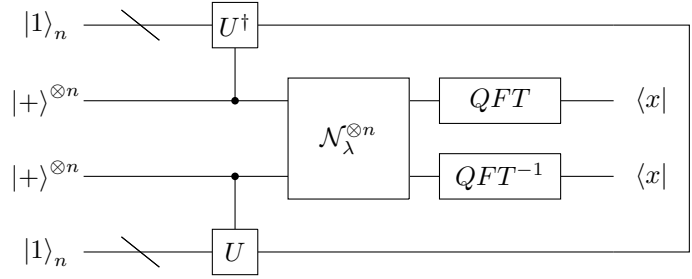
\begin{center}
    \begin{tikzpicture}
        \begin{scope}
        \foreach \i in {0,1}{
        \draw(-1,-1*\i)--(6,-1*\i);
        \draw(-.5,.2)--(0,-.2);
        }

        \node at (-1.5,-1){$\ket{+}^{\otimes n}$};
        \node at (-1.5,0){$\ket{1}_n$};
        
        \fill(1,-1) circle(.05);
        \draw(1,0)--(1,-1);
        \smallbox{1}{0}{0.3}{0.3}
        \node at (1,0){$U^{\dagger}$};

        \smallbox{4.5}{-1}{0.8}{0.3}
        \node at (4.5,-1){$QFT$};

        \node at (6.5,-1) {$\bra{x}$};

        \end{scope}

        \begin{scope}[shift={(0,-3)},yscale=-1]
        \foreach \i in {0,1}{
        \draw(-1,-1*\i)--(6,-1*\i);
        \draw(-.5,-.2)--(0,.2);
        }

        \node at (-1.5,-1){$\ket{+}^{\otimes n}$};
        \node at (-1.5,0){$\ket{1}_n$};
        
       \fill(1,-1) circle(.05);
        \draw(1,0)--(1,-1);
        \smallbox{1}{0}{0.3}{0.3}
        \node at (1,0){$U$};

        \smallbox{4.5}{-1}{0.8}{0.3}
        \node at (4.5,-1){$QFT^{-1}$};

        \node at (6.5,-1) {$\bra{x}$};
        
        \end{scope}
        
        \smallbox{2.5}{-1.5}{.8}{.8}
        \node at (2.5,-1.5){$\mathcal{N}_{\lambda}^{\otimes n}$};

        \draw (6,0)--(7,0)--(7,-3)--(6,-3);
    \end{tikzpicture}
\end{center}
Here, we applied that $(QFT^{-1})^{\dagger}=QFT$.

\begin{proposition}
For the Shor circuit with single-qubit depolarizing noise rate
$\lambda$, the probability of measuring $x$ is
\begin{equation}\label{Equ: P=sum h1}
P_{\lambda}(x)= 2^{-2n} \sum_{0\leq u,v < 2^n ; r \mid u-v}
(1-\lambda)^{h(u,v)} \omega^{x (u-v)}  ,
\end{equation}
where $h(u,v)$ is the Hamming distance between the binary expansions
of $u$ and $v$.
\end{proposition}

\begin{proof}
We analyze $P_{\lambda}(x)$ through the Pauli path expansion.
We represent the depolarizing channel $\mathcal{N}_{\lambda}$ in terms of the Pauli basis:

\begin{center}
    \begin{tikzpicture}
        \begin{scope}
            \draw(0,0)--(2,0);
            \draw(0,-1)--(2,-1);
            
        \smallbox{1}{-.5}{.8}{.8}
        \node at (1,-.5){$\mathcal{N}_{\lambda}$};
        \end{scope}

        \begin{scope}[shift={(3,0)}]
        \node at (-.5,-.5){$=\displaystyle \frac{1}{2}$};
            \draw(0,0)--(.6,0)--(.6,-1)--(0,-1);
            \draw(2,0)--(1.4,0)--(1.4,-1)--(2,-1);
            
        \end{scope}

         \begin{scope}[shift={(8.5,0)}]
        \node at (-1.5,-.5){$+\sum\limits_{P=X,Y,Z} \displaystyle \frac{1-\lambda}{2}$};
            \draw(0,0)--(.6,0)--(.6,-1)--(0,-1);
            \draw(2,0)--(1.4,0)--(1.4,-1)--(2,-1);
            
        \smallbox{.6}{-.5}{0.25}{.25}
        \smallbox{1.4}{-.5}{0.25}{.25}
        
        \node at (.6,-.5){$P$};
         \node at (1.4,-.5){$P$};
        \end{scope}

    \end{tikzpicture}
\end{center}
and

\begin{center}
    \begin{tikzpicture}
        \begin{scope}
            \draw(0,0)--(2,0);
            \draw(0,-1)--(2,-1);
            
        \smallbox{1}{-.5}{.8}{.8}
        \node at (1,-.5){$\mathcal{N}_{\lambda}^{\otimes n}$};
        \end{scope}

         \begin{scope}[shift={(4.5,0)}]
        \node at (-1.5,-.5){$=\sum\limits_{\vec{P}} \displaystyle \frac{(1-\lambda)^{h(\vec{P})}}{2^n}$};
            \draw(0,0)--(.6,0)--(.6,-1)--(0,-1);
            \draw(2,0)--(1.4,0)--(1.4,-1)--(2,-1);
            
        \smallbox{.6}{-.5}{0.25}{.25}
        \smallbox{1.4}{-.5}{0.25}{.25}
        
        \node at (.6,-.5){$P$};
         \node at (1.4,-.5){$P$};
         
         \node at (2.5,-.6) {$.$};
        \end{scope}

    \end{tikzpicture}
\end{center}
The sum is over all $n$-qubit Pauli words
$\vec{P}=\bigotimes_{\ell=0}^{n-1}P_\ell$, where
$P_\ell\in\{I,X,Y,Z\}$, and
$h(\vec{P})=\#\{\ell:P_\ell\neq I\}$. The expansion follows from
\begin{equation}
\mathcal{N}_\lambda(I)=I,
\qquad
\mathcal{N}_\lambda(P)=(1-\lambda)P
\quad(P=X,Y,Z).
\end{equation}

For a formal evaluation of the diagram, let
\begin{equation}
M_x:=QFT\ket{x}\!\bra{x}QFT^{-1}.
\end{equation}
The Fourier state factors as
\begin{equation}
QFT\ket{x}=\frac{1}{\sqrt{2^n}}
\bigotimes_{\ell=n-1}^{0}
(\ket{0}+\omega^{x2^{\ell}}\ket{1}),
\end{equation}
and hence
\begin{equation}
M_x=\frac{1}{2^n}\bigotimes_{\ell=n-1}^{0}
\begin{bmatrix}
1 & \omega^{-x2^{\ell}} \\
\omega^{x2^{\ell}} & 1
\end{bmatrix}.
\end{equation}
On each tensor factor, the diagonal matrix elements are unchanged by
$\mathcal{N}_\lambda$, while the off-diagonal elements are multiplied
by $1-\lambda$. Therefore
\begin{equation}
\bra{v}\mathcal{N}_\lambda^{\otimes n}(M_x)\ket{u}
=2^{-n}(1-\lambda)^{h(u,v)}\omega^{x(v-u)}.
\end{equation}

The channel $\mathcal{N}_\lambda$ is self-adjoint with respect to the
Hilbert--Schmidt inner product. Using the control-register state
$\rho_r$ from the previous section, we obtain
\begin{align*}
P_\lambda(x)
&=\operatorname{Tr}\!\left(
M_x\mathcal{N}_\lambda^{\otimes n}(\rho_r)\right)\\
&=\operatorname{Tr}\!\left(
\mathcal{N}_\lambda^{\otimes n}(M_x)\rho_r\right)\\
&=2^{-2n}\sum_{0\leq u,v<2^n;\ r\mid u-v}
(1-\lambda)^{h(u,v)}\omega^{x(v-u)}\\
&=2^{-2n}\sum_{0\leq u,v<2^n;\ r\mid u-v}
(1-\lambda)^{h(u,v)}\omega^{x(u-v)},
\end{align*}
where the last line follows by exchanging $u$ and $v$.
\end{proof}

Each matching Pauli word gives a path contributing to $P(x)$. In the ideal circuit
all compatible paths retain their full weight resonance at an effective $x$. The factor $(1-\lambda)^{h(u,v)}=(1-\lambda)^{h(\vec{P})}$ shows
that depolarizing noise contracts a path once for every non-identity
Pauli factor, which makes this resonance fragile.

Note that $h(u,v)\leq h(u)+h(v)$, for the Hamming weight of $u$ and $v$ in binary. 
There is a unique pair $(u',v')$ such that
\begin{equation}
u_i-v_i=u'_i-v'_i, ~\forall~0\leq i<n ;
\end{equation}
\begin{equation}
h(u')+h(v')=h(u',v')=h(u,v).
\end{equation}
The bit strings of $u'$ and $v'$ are different at $h(u,v)$ positions, and zero elsewhere. 
We call $(u',v')$ the Minimal Weight Representative of $(u,v)$.
There are exactly $2^{n-h(u,v)}$ pairs with the same minimal weight representative $(u',v')$ for all possible choices of bit strings at the rest $n-h(u,v)$ positions.
Therefore
\begin{equation}\label{Equ: P=sum h2}
P_{\lambda}(x)=
2^{-2n} \sum_{0\leq u,v < 2^n ; r \mid u-v; h(u)+h(v)=h(u,v)}
\omega^{x (u-v)}  (1-\lambda)^{h(u,v)} 2^{n-h(u,v)}.
\end{equation}
 
Denote the set of minimal weight representatives with Hamming weight at most $d$ by
\begin{equation}
MWR(d):=\{(u,v): 0\leq u,v < 2^n, r \mid u-v, h(u)+h(v)=h(u,v), h(u,v)\leq d\},
\end{equation}
with cardinality $L(d)$.

\begin{theorem}
For the Shor circuit with single-qubit depolarizing noise rate
$\lambda$ and any integer $d\geq 0$, the probability of outcome $x$ is
bounded as
\begin{equation}
P_{\lambda}(x)\leq \frac{L(d)}{2^{n}}
+(1-\lambda)^{d+1}\left(\frac{1}{r}+\frac{1}{2^n}\right).
\end{equation}
Here $L(d)$ is the number of minimal representative pairs $(u,v)$ such that
$0\leq u,v < 2^n$, $r \mid u-v$ with Hamming weight
$h(u,v)\leq d$.
\end{theorem}

\begin{proof}
In Eq.~\eqref{Equ: P=sum h2}, there are $L(d)$ terms in the sum with low weight $h(u,v)\leq d$ and each term is bounded by $2^{-n}$.
So the sum of these low weight terms is bounded by $L(d) 2^{-n}$. 
In Eq.~\eqref{Equ: P=sum h1},
every term with high weight $h(u,v)\geq d+1$ is bounded in
absolute value by $2^{-2n}(1-\lambda)^{d+1}$. The total number of
terms is bounded by $2^{n}(\frac{2^n}{r}+1)$. Hence the absolute
value of the sum of the high-weight terms is bounded by
$(1-\lambda)^{d+1}\left(\frac{1}{r}+\frac{1}{2^n}\right)$.
Therefore, $P_{\lambda}(x)$ is bounded by the sum of the two bounds.
\end{proof}

There is also a simple bound independent of $r$:
\begin{equation}\label{Equ: L combinatorial bound}
L(d)\leq \sum_{j=0}^{d}2^j\binom{n}{j}=O(n^d)
\end{equation}
for every fixed $d$. Indeed, a minimal representative of weight $j$
is specified by choosing its $j$ nonzero positions and, at every such
position, choosing whether the nonzero bit belongs to $u$ or to $v$.
The additional condition $r\mid u-v$ can only reduce this count.

To state the consequence for the usual resonance peaks precisely,
define
\begin{equation}
\mathcal{E}_r:=\left\{x\in\{0,\ldots,2^n-1\}\ \middle|\
\exists k\in\{0,\ldots,r-1\},\
\left|\frac{x}{2^n}-\frac{k}{r}\right|
\leq\frac{1}{2^{n+1}}\right\}.
\end{equation}
There are at most $2r$ outcomes in $\mathcal{E}_r$. Hence the theorem,
$r<N$, and $2^n>N^2$ imply
\begin{align}
P_{\lambda}(\mathcal{E}_r)
&=\sum_{x\in\mathcal{E}_r}P_{\lambda}(x) \notag\\
&\leq 2r\left(
\frac{L(d)}{2^{n}}+(1-\lambda)^{d+1}
\left(\frac{1}{r}+\frac{1}{2^n}\right)\right) \notag\\
&\leq \frac{2(L(d)+1)}{2^{n/2}}+2(1-\lambda)^d.
\label{Equ: noisy resonance bound}
\end{align}
For every fixed noise rate $0<\lambda\leq1$,
Eq.~\eqref{Equ: L combinatorial bound} and
Eq.~\eqref{Equ: noisy resonance bound} give
\begin{equation}
\lim_{d\to\infty} \lim_{N\to\infty}P_{\lambda}(\mathcal{E}_r)=0.
\end{equation}
Indeed, for each fixed $d$ the first term vanishes as $N$ grows, and
then $d$ can be chosen large enough to make $(1-\lambda)^d$
arbitrarily small. Thus the constant lower bound in
Eq.~\eqref{Equ:Shor lower bound} is destroyed by any fixed positive
noise rate in this one-layer model.

Now for a fixed $n$, let us discuss what the bound $L(d)$ of lower frequency terms says about Shor's factoring algorithm. For a fixed $d$, denote the set of binary numbers with
Hamming weight at most $d$ by
\begin{equation}
W_d:=\{u~:~0\leq u <2^n, h(u)\leq d\}.
\end{equation}
The set $W_d$ has $O(n^d)$ elements. For each $u\in W_d$, the
residue $a^u\bmod N$ can be computed by repeated squaring. 
In this way all $L(d)$ elements of $MWR(d)$ can be listed in polynomial time $O(n^d)$.
They will produce all $L(d)$ pairs of $(u,v)$ in $MWR(d)$.
For each pair $(u_i,v_i)$ in $MWR(d)$, denote their difference as $|u_i-v_i|=\Delta_i=rq_i$. 
Then
\begin{equation}
\gcd(\Delta_1,\ldots,\Delta_m)
=r\gcd(q_1,\ldots,q_m).
\end{equation}
If the quotients $q_i$ behaved like independent random positive
integers, the asymptotic probability that their gcd is one would be
\begin{equation}
\frac{1}{\zeta(m)}
=\left(\sum_{k=1}^{\infty}\frac{1}{k^m}\right)^{-1},
\end{equation}
which equals $6/\pi^2$ for $m=2$ and approaches one exponentially in
$m$. This is only a heuristic argument which does not guaranty
that the gcd is exactly the order $r$.

In summary, the measurement distribution under one layer of
independent depolarizing noise splits into a low-Hamming-weight sector
whose modular relations can be enumerated classically for fixed $d$
and a high-weight sector suppressed by $(1-\lambda)^d$. This result is
specific to the near-resonant event $\mathcal{E}_r$ and to the stated
single-layer noise model. It does not prove an efficient classical
simulation of the full noisy circuit, exclude other successful
post-processing events, or rule out a fault-tolerant quantum
advantage. With quantum error correction the effective logical noise
rate may depend on $N$ and approach zero; that scaling and its overhead
must be analyzed separately.

\section*{Acknowledgments}

Z.L. acknowledges support from the Beijing Natural Science Foundation (Grant No. Z220002).

\bibliographystyle{apsrev4-2}
\bibliography{references}

\end{document}